\def\comments{1}
\documentclass{article}
\usepackage{graphicx} 
\usepackage{preamble}
\title{Hardening Confidential Federated Compute\\ against Side-channel Attacks}
\author{James Bell-Clark, Albert Cheu, Adria Gascon, Jonathan Katz\thanks{All authors employed by Google.}}
\date{}

\begin{document}

\maketitle
\begin{abstract}
    In this work, we identify a set of side-channels in our Confidential Federated Compute platform that a hypothetical insider could exploit to circumvent differential privacy (DP) guarantees. We show how DP can mitigate two of the side-channels, one of which has been implemented in our open-source library.
\end{abstract}

\section{Background} 
In the central model of differential privacy (DP), a single party is trusted to apply a DP algorithm on data contributed by individuals. The raw data is assumed to be accessible only by that party; in the lower left of Figure \ref{fig:central-vs-cfc}, this is represented by a locked arrow. Note that peer uploads, represented in red, could be maliciously crafted but DP protection accounts for this event \cite{DMNS06}.

Distributed models of DP computation handle settings with less trust, but face a number of challenges. The primary one studied is worsened privacy-accuracy tradeoff: when answering the example summation query in the top left of Figure \ref{fig:central-vs-cfc}, local DP protocols would have error growing with the number of client contributions \cite{BNO08,CSS12}, while the best error in the central model is independent of that value. Entities like a trusted shuffler and secure aggregation enable protocols that match the error in the central model \cite{CSUZZ19, BC20, CZ22, KLS21, BKMGB22}, but these results do not hold for harder problems \cite{CU21, Cheu21}. Moreover, communication complexity lower bounds exist \cite{COCB22,CCKS22}. A less-obvious challenge with distributed DP is that the inputs to the DP protocol are limited to what can be computed by an individual's device. Insights that can be obtained from large models, for example, are ruled out.

In prior work, we describe \emph{confidential federated compute (CFC)}, a platform for federated learning and analytics that strives to overcome the trust-utility tradeoff and support model inference that cannot be done on-device \cite{CFC,KMS}. Though CFC can be restricted to DP workloads, this work identifies side-channels in CFC---and any sufficiently similar system---that could be used to recover sensitive inputs. We also describe how DP noise can mitigate the leakage. We limit our scope to GROUP BY SUM queries.

We now review CFC. As its name suggests, \emph{confidentiality} is a primary objective: we aim to limit visibility of uploaded client data to server-side actors. Naturally, we rely on DP algorithms: contribution bounding, aggregation, and noise addition are all steps we have implemented \cite{aggCoreRepo}. We run these algorithms inside  \emph{confidential virtual machines (CVMs)}.\footnote{CVMs are sometimes called Trusted Execution Environments (TEEs), a term that has also been applied to technologies like Intel's SGX. We prefer the more precise term.} Designed to isolate guest memory from the host that launches it, CVMs are supported by certain lines of hardware and each line has their own techniques for isolation. For example, AMD's EPYC processors facilitate different levels of Secure Encrypted Virtualization (SEV) \cite{sev}; we use secure nested paging (SNP) \cite{sevsnp}. As we discuss, the isolation that existing CVMs offer is not complete, so our goal is to find the limitations and then overcome them.

In addition to confidentiality, CFC offers \emph{external verifiability}: we provably limit computation on client data to executables compiled from open-source code. For this purpose, we develop a stateful entity called the key management service (KMS). Uploads from client devices to our platform are encrypted with keys that the KMS manages. The KMS is given a policy, e.g., ``client data for this pipeline can be processed once by an $\eps \leq 1$-DP computation and no more.'' Any CVM that wishes to decrypt either a client upload or an intermediate result must request a key from the KMS. The KMS checks the code running on the CVM against the policy via remote attestation, which CVMs support. The KMS itself runs in a CVM; when an outsider verifies correctness of the KMS, they also confirm the privacy budget. Implementation details of the KMS can be found in our prior work and associated Github repository~\cite{KMS,cfcRepo}.

\emph{Parallelism} enables our platform to efficiently perform large-scale data analytics, as visualized on the right side of Figure \ref{fig:central-vs-cfc}. We launch multiple leaf workers, each a CVM computing a partial histogram on its partition of all contributions. A root CVM gathers what was scattered---the serialized intermediate histograms---and applies noise to the final histogram. We remark that the hosts of the CVMs are not externally verified to adhere to any policy, so we color them red. We also note that the KMS is not the focus of this work, so that component is omitted from the diagram for clarity.

\begin{figure}
    \centering
    \includegraphics[width=0.9\textwidth]{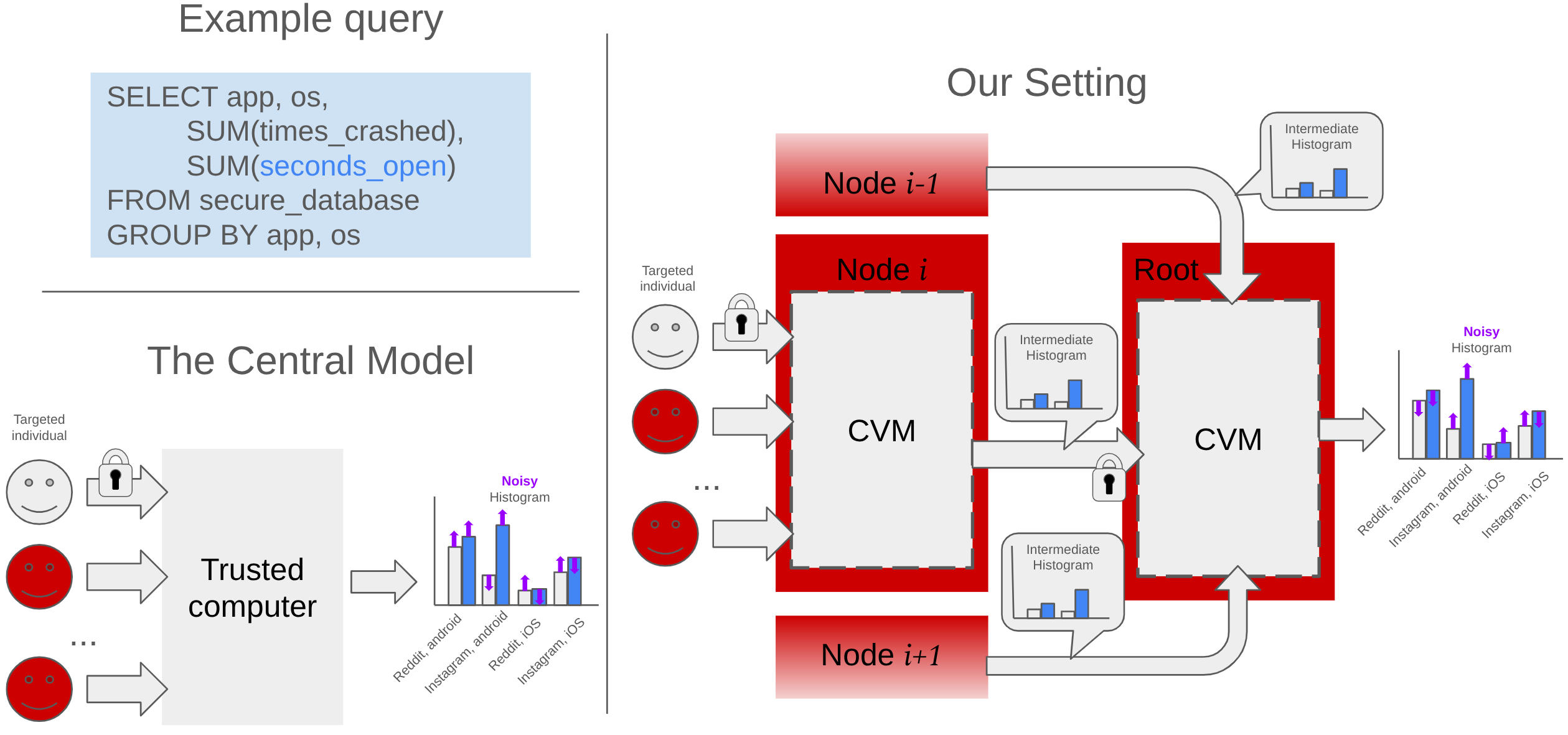}
    \caption{Visualization of the central model contrasted with our system, in the context of a toy SQL query.  Red objects are adversarial. Locked arrows indicate encrypted communications. }
    \label{fig:central-vs-cfc}
\end{figure}

Although external verifiability brings our platform close to the ideal central model, we find that \emph{certain side channels are inevitable in any system that uses multiple existing CVMs} like ours. Compromised operators could use these to side-step DP guarantees. 

\subsection{Side channels and how to observe them}

\begin{table}
\centering
\begin{tabular}{l|l|l}
Side-channel & Description & Note \\ \hline
Message length      & Length of plaintext sent from one worker to another     & This work            \\
Memory allocation              & How much memory is in use by CVM & This work            \\ \hline
Message timing      & When one CVM communicates with another                                & Open research          \\
Data access pattern                    & Transcript of where CVM read and/or wrote            & Open research           \\
Code access pattern                    & Transcript of where CVM got instructions         & Open research           \\
Data change pattern            & Transcript of changes to data made by a CVM           & Open research         
\end{tabular}
\caption{Overview of side channels that could be used in a privacy attack.}
\label{tab:side-channels}
\end{table}

We overview key side channels in Table \ref{tab:side-channels}. Some are easier for an adversary to observe than others. For the simplest example, consider an encrypted message transmitted by one of the workers, corresponding to the locked arrow linking from node $i$ to the root in Figure~\ref{fig:central-vs-cfc}. In terms of our DP algorithms library, the payload is generated by the \texttt{Serialize} method of our \texttt{DPGroupByAggregator} class \cite{aggCoreRepo}. Though the payload is not directly accessible due to encryption, the length of the payload and the timing of its transmission are observable to anyone monitoring network traffic.

A more challenging side-channel to obtain is the amount of memory a CVM uses. This requires an adversary that has remote access to the host machine, but interference with the CVM is unneccessary: because tools for monitoring virtual machines already exist in the host. For example, performance-tracking commands such as \texttt{perf stat} can be used to track page faults \cite{perfstat}, which are positively correlated with (new) memory allocation.

To surface other side channels, current state-of-the-art requires an attacker to have root access to the host of a CVM. But once that is achieved, they can deploy vendor-specific attacks. For AMD SEV-SNP the framework called SNPeek demonstrates it is possible to, among other things, force page faults whenever a page is requested \cite{SNPeek}. Each page fault essentially carries a flag indicating if the affiliated page contains data or code. Furthermore, SNPeek can surface when the ciphertext in a memory location changes, which is problematic since SEV-SNP applies deterministic encryption (so plaintext that does not change leads to no ciphertext change).


\subsection{Two Concrete Attacks}
Here, we show how an adversary can exploit two of the side-channels to launch a privacy attack.


\paragraph{Attack using Message Length} We begin with the observation that message length is correlated with the input. To illustrate why, suppose we are answering the example SQL query in Figure \ref{fig:central-vs-cfc}. Suppose also that Sybil users exclusively contribute to groups 
$\{$ (Reddit, Android), (Instagram, Android), (X, Android), (TikTok, Android), (Youtube, Android) $\}$ and an instance of our \texttt{DPGroupByAggregator} class ingests these inputs.
In preliminary versions,
the output of the \texttt{Serialize()} method would be 190 bytes long after ingesting another copy of (Reddit, Android) but it would be 217 bytes long if the last input were instead (Reddit, iOS). Appendix \ref{apdx:msg-length-attack} explains how we derived these values. Intuitively, the intermediate histogram only has the contributions to observed groups, which changes depending on the input.

\paragraph{Attack using Memory Allocation} Dynamically-sized data structures change the amount of memory allocated based on the data. In the running example, a dynamic hash table pairing each group with a running sum will use more space when there are more distinct groups. Off-the-shelf implementations like C++'s \texttt{std::unordered\_map} have a deterministic resize schedule: each implementation has a fixed sequence of capacities such that, if we repeatedly add key-value pairs, the amount of reserved memory will grow when the number of distinct keys reaches each capacity.

With this in mind, the adversary can craft Sybil uploads that consist of $T$ groups just as in the previous attack, this time choosing $T$ to be just below the point of a resize. \texttt{perf stat} will report measurably more page faults when a resize happens than not, serving as a signal that the target user is contributing new items. Refer to Figure \ref{fig:perf_attack} for a visualization of the spike in page faults.

\subsection{Our mitigations}
We show that variants of existing techniques can be used to mitigate the  threats above without introducing significant overhead. For the message-length attack, we have implemented a simple defense: \emph{add padding to inter-server messages}, where the padding length is determined by a shifted and clamped version of the Laplace distribution. The scale of noise is proportional to a sensitivity calculated by a bespoke algorithm that factors in contribution bounds and data types. 
In Figure \ref{fig:msg-mitigation}, we show that the overhead of the padding halves every time we double the number of groups. 
To generate the plot, we fixed $\delta=10^{-4}$ and bounded contributions to 1 group per client. Our \texttt{Serialize()} method was run $40$ times to generate each box. The orange lines represent medians, while the boxes represent quartiles. More details can be found in Appendix \ref{apdx:msg-length-mitigation}.

\begin{wrapfigure}[16]{r}{0.5\textwidth}
    \includegraphics[width=\linewidth]{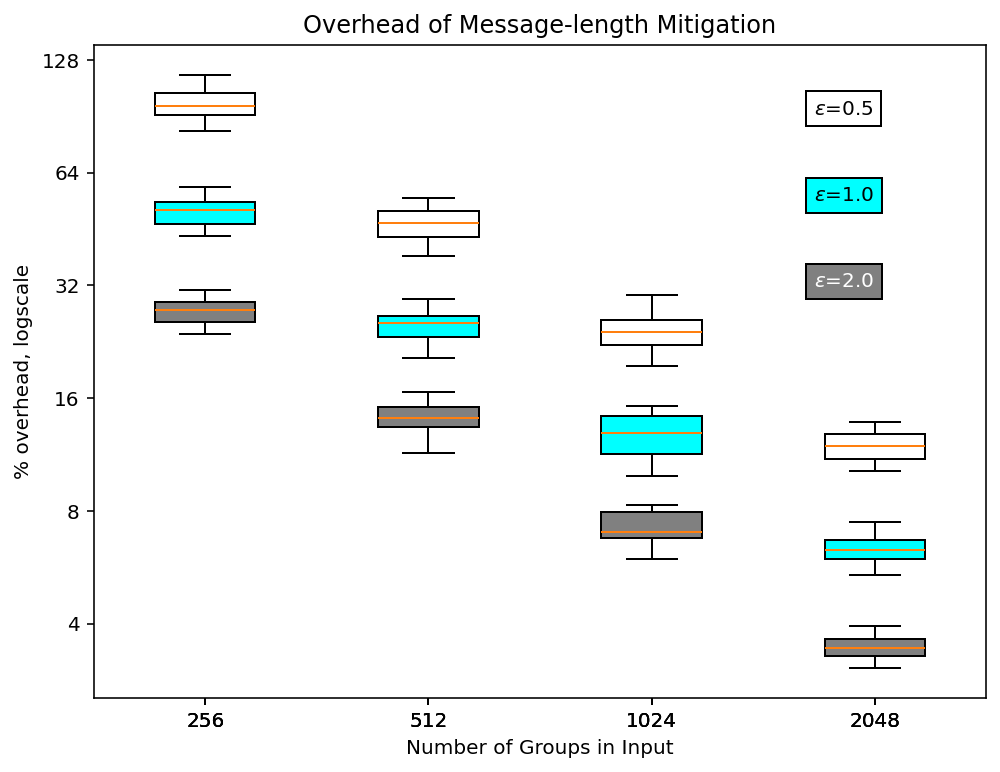}
    \caption{Overhead due to padding. Color distinguishes $\eps$ values.  }
    \label{fig:msg-mitigation}
\end{wrapfigure}

For the memory-allocation attack, we propose a variant of the AboveThreshold mechanism  to \emph{randomize when data structures resize} \cite{DR14}. Specifically, we randomize the current capacity of the data structure and we compare it with a noisy estimate of the data structure's load after every insertion, resizing only when the threshold is crossed. Care is taken to ensure capacity is never exceeded. We also observe that changing one contribution can only cause load to move in one direction, allowing us to reduce noise in AboveThreshold.

\subsection{Related Work}
DP has been used to plug side-channel leakage in prior systems like Hermetic by Xu et al.\ and Menhir by Reichert et al.~\cite{hermetic, menhir}. DP padding has also been studied in more theoretical work such as Allen et al.\ and $\varepsilon$psolute by Bogatov et al.~\cite{allen2019,bogatov2021}. But $\varepsilon$psolute does not use CVMs (or TEEs more generally), instead focusing on outsourced database systems.
Ratliff and Vadhan have a line of theoretical work on mitigating timing attacks, where random delays are injected to mask the influence of any single contribution \cite{RatliffV24, RatliffV25}; like our message lengths, these can only be positive.

Retrieval of confidential records is an important use case for the authors of $\varepsilon$psolute and Menhir; the number of retrieved records is data-dependent, so randomly padding the volume is a natural solution. In our case, we are not directly returning records but intermediate state emitted by a worker consists of histograms from the records; for similar reasons, we pad the length of that state.
Meanwhile, Hermetic relies upon bespoke ``oblivious execution environments,'' which are engineered to plug side-channel leakage for a small set of primitive functions. But these depend on a custom hypervisor; in our setting, we make no assumption about the trustworthiness that component.

\subsection{Open Research Directions}
This work focuses on defenses of two specific side-channel attacks. Future work may document more side-channels, as well as expand our set of defenses. For example, the Oblivious RAM (ORAM) literature offers options to mitigate leakage of the data-access pattern, at the price of time and memory overhead \cite{futorama,pathORAM,ringORAM}. Naively, one could allocate as much ORAM as one would ever need prior to reading any input. But a more dynamic approach would start with a small ORAM and repeatedly apply AboveThreshold to detect when to migrate the contents to a bigger ORAM. We would only pay ORAM overhead proportional to the content of an input stream, not its length.

More advanced methods would exploit the write-then-read pattern to improve efficiency. Specifically, read-from-memory operations can be deferred when answering GROUP-BY SUMs. Off-the-shelf ORAMs like PathORAM or RingORAM do not optimize for this write-then-read pattern.

\section{Preliminaries}
\subsection{Differential Privacy}
Two datasets are neighboring if they differ by changing the contributions of one individual. This is the replacement neighboring relation. The alternative add-remove relation is ill-suited for our threat model: because the adversary controls what is input to the workers, membership in the input is already known.

An algorithm $M$ is $(\eps,\delta)$-differentially private or DP if, for any neighboring inputs and any outcome $S$,
\begin{align*}
\pr{}{M(D)\in S} &\leq e^\eps \cdot \pr{}{M(D')\in S}+\delta\\
\pr{}{M(D')\in S} &\leq e^\eps \cdot \pr{}{M(D)\in S}+\delta
\end{align*}
which we sometimes shorthand as $\pr{}{M(D)\in S} \approx_{\eps,\delta} \pr{}{M(D')\in S}$. We drop $\delta$ if it is equal to 0.

A common way to argue that an algorithm is $(\eps,\delta)$-DP is to show that it is close to an $\eps$-DP algorithm, regardless of input. We use $\tvd{P-Q}$ to denote the total variation distance between distributions $P,Q$.
\begin{lem}
\label{lem:close-to-pure}
Suppose $M$ is $\eps$-DP. If there is another algorithm $M'$ that satisfies $\tvd{M(D)-M'(D)}\leq \frac{\delta}{e^\eps+1}$ for any input $D$, then $M'$ is $(\eps,\delta)$-DP.
\end{lem}
\begin{proof}
For any $S$ and neighboring pair $D,D'$,
\begin{align*}
&\pr{}{M'(D) \in S}\\
\leq{}& \pr{}{M(D)\in S}+\delta/(1+e^\eps) \\
\leq{}& e^\eps \pr{}{M(D')\in S}+\delta/(1+e^\eps) \\
\leq{}& e^\eps (\pr{}{M'(D')\in S} + \delta/(1+e^\eps)) + \delta/(1+e^\eps) \\
={}& e^\eps \pr{}{M'(D')\in S} + \delta \qedhere
\end{align*}
\end{proof}

\section{Defense against Message Length Adversary}

The main idea here is simple: calculate the sensitivity of the serialized string's length, then add padding characters with a random length calibrated to the sensitivity. We refer the interested reader to the \texttt{Serialize} function of our \texttt{GroupByAggregator} and \texttt{DPGroupByAggregator} classes.

The sensitivity of \texttt{Serialize} is a function of the customer-provided SQL query and contribution bounds. Our system prevents one contributor from influencing more than $\Delta$ groups,\footnote{In code, this bound is named \texttt{max\_groups\_contributed} and is required to construct a \texttt{DPGroupByAggregator}.} which means we can trace how that influence trickles down to the length of \texttt{Serialize}'s output. The intermediate histogram is specified by columns corresponding to the grouping keys given in the query and columns corresponding to the aggregations; because our contribution bound is $\Delta$, we have a limit how much these columns can be lengthened. Each element of an aggregation column has a fixed size by virtue of being a 32-bit or 64-bit number but grouping columns can have string-typed values. Therefore, sensitivity bounds can only be derived when we impose a hard limit on how long these strings can get.

A column of strings is encoded differently than a column of numbers. Specifically, we record the lengths of the strings as a sequence of numbers followed by the strings themselves but a column of numbers simply gets mapped to the bytes that represent the numbers. We account for this distinction in our sensitivity calculations. We also account for the fact that our data is packaged within protocol buffers: we use information available at https://protobuf.dev/programming-guides/encoding/ to compute the sensitivity of variable-length integers (``varints'') that encode the number of bytes consumed by each string and each column encoding. The interested reader can refer to \texttt{CalculateSerializeSensitivity} for more information.

To generate the length of padding, we lightly modify the Laplace mechanism to ensure it is never reports a value less than its input.

\begin{algorithm}
\caption{The PositiveLaplaceMechanism}
\textbf{Parameters}: $\eps,\tau$, positive real values

\KwIn{$v$, with sensitivity $\Delta$}
\KwOut{A value at least as large as $v$}

Sample $\eta\sim \Lap(\Delta/\eps)$

$\tilde{v}\gets v+\tau+\eta$

\Return{$\max(v, \tilde{v})$}
\end{algorithm}

To prove that the PositiveLaplaceMechanism is $(\eps,\delta)$-DP, we can invoke Lemma \ref{lem:close-to-pure} for a sufficiently large $\tau$ parameter. Let $q_\phi$ be the absolute value of the $\phi$-quantile of the Laplace distribution used in the PositiveLaplaceMechanism.

\begin{thm}
\label{thm:simple-positive-laplace}
Let $h:= \delta/(1+e^\eps)$. If $\tau := q_h$ then PositiveLaplaceMechanism is $(\eps, \delta)$-DP.
\end{thm}
\begin{proof}
Let $M$ denote the mechanism that is identical to PositiveLaplaceMechanism except it directly returns $\tilde{v}$. By the choice of $\tau$ and $h$, PositiveLaplaceMechanism$(v)$ is within $h$ of $M(v)$ in total variation distance for every $v$. Also note that $M$ is just the textbook Laplace mechanism shifted by $\tau$, which does not impact sensitivity. Thus Lemma \ref{lem:close-to-pure} implies $(\eps, \delta)$-DP.
\end{proof}

A more bespoke analysis that shows a slightly smaller $\tau$ setting suffices, which our implementation uses. We defer the details to Appendix \ref{apdx:positive-laplace}.

\section{Defense against Memory Allocation Adversary}

In this section, we describe associative maps that resize in a differentially private manner, negating the threat posed by adversaries that snoop on the amount of memory in use. We focus on associative maps because they are well-suited for federated analytics queries, our motivating problem.

Our starting point is the AboveThreshold algorithm, a classic tool from the DP literature. Given a dataset $D$ and a threshold $T$ it allows one to answer queries of the form ``is $f_i(D)>T$?'' in a differentially private manner, so long as the sensitivity of each $f_i$ is bounded. Although the formulation aligns well with our resizing objective---let $f_i$ be queries on distinct elements and set $T$ to the size threshold---we will modify it to enforce a \emph{strict stopping condition}: the first time that the threshold is reached, the new algorithm is guaranteed to halt. This prevents insertions to an at-capacity associative map. We additionally identify a \emph{unidirectional sensitivity} property which we use to reduce the scale of noise, preventing overly premature resizes. The property is also present when applying AboveThreshold to private quantile estimation.

Note that we assume that each contributor can only contribute $\Delta=1$ key-value pair; parameters can be up-scaled when contributions come in larger batches.







\subsection{Tuning AboveThreshold for Unidirectional Sensitivity (UDS)}

Consider the following applications of AboveThreshold:
\begin{itemize}
    \item Quantiles: when every $f_i(D)$ has the form ``how many elements $D$ have value $\leq v$?'' for an increasing sequence of $v$ and $T=0.9|D|$, AboveThreshold halts at an estimate of the 90th percentile of $D$.
    \item Load estimation: when each $f_i$ has the form ``how many distinct elements make up the first $i$ elements of $D$?'' and $T$ measures a datastructure's capacity for distinct elements, AboveThreshold will signal when the data structure would be close-to or just-above capacity. 
\end{itemize}
In both, we observe that not only do these queries have sensitivity 1, they have what we call unidirectional sensitivity: if $f_i(D) > f_i(D')$ for some $i$ and neighboring datasets $D,D'$ then there is no $j$ where $f_j(D) < f_j(D')$.\footnote{For the load estimation application, notice that when we replace one element of $D$, we cannot find two points in time where the number of distinct elements move in opposite directions. For the quantile application, when we replace one element in the database, we also cannot find two below-$v$ counts that move in opposite directions} AboveThreshold was designed for the more general scenario where $f_i$ may decrease or increase arbitrarily; modifying the algorithm for unidirectionality will improve accuracy.

Specifically, we present a more precisely calibrated version of AboveThreshold that reduces the scale of noise per query by a factor of two ($4/\eps$ to $2/\eps$).
\begin{algorithm}
\caption{UDSAboveThreshold, AboveThreshold modified for UDS queries}
\label{alg:uds-abovethreshold}
\KwIn{Dataset $D$, 1-sensitive queries $\{f_i\}$, threshold $T$ privacy budget $\eps$}
\KwOut{Stream $a_1,\dots$}

Let $\hat{T}\gets T+\Lap(2/\eps)$

\For{$i\in \N$}{
    Let $\nu_i \gets \Lap(2/\eps)$ \tcc{In Dwork \& Roth's AboveThreshold, it is $4/\eps$}

    \If {$f_i(D) +\nu_i \geq \hat{T}$}{
        Output $a_i \gets \top$

        Halt
    }
    \Else{
        Output $a_i \gets \bot$
    }
}
\end{algorithm}

\begin{thm}
\label{thm:uds}
For functions $\{f_i\}$ that are both 1-sensitive and have unidirectional sensitivity, UDSAboveThreshold satisfies $\eps$-DP.
\end{thm}

We defer the proof to Appendix \ref{apdx:uds}. At a high level the proof is the same as the textbook proof of AboveThreshold, except we use unidirectionality to halve the width of a sensitive interval and therefore halve the scale of noise.

\subsection{Enforcing a Strict Stopping Condition on UDSAboveThreshold}

\begin{algorithm}
\caption{StrictUDSAboveThreshold, a variant that never proceeds once $T$ is crossed}
\label{alg:strict-uds-abovethreshold}
\KwIn{$D,\{f_i\}, T, \eps, \delta$}
\KwOut{Stream $a_1,\dots$}

Let $q$ denote the $1-\frac{\delta}{2(1+e^\eps)}$ quantile of $\Lap(2/\eps)$

Let $\hat{T}\gets T+\Lap(2/\eps) - 2q$

\For{$i\in \N$}{
    Let $\nu_i \gets \Lap(2/\eps)$

    \If {$f_i(D)\geq T$ or $f_i(D) +\nu_i \geq \hat{T}$}{
        Output $a_i \gets \top$

        Halt
    }
    \Else{
        Output $a_i \gets \bot$
    }
}
\end{algorithm}

Algorithm \ref{alg:strict-uds-abovethreshold} has two differences with Algorithm \ref{alg:uds-abovethreshold}: the threshold is shifted down by a quantile of the Laplace distribution and we change the halting condition to guarantee that the algorithm never proceeds once $T$ is crossed.

\begin{thm}
For functions with unidirectional sensitivity, StrictUDSAboveThreshold satisfies $(\eps,\delta)$-DP.
\end{thm}
\begin{proof}
Let $M$ denote the intermediate mechanism that shifts the threshold but does not change the halting condition. Notice that changing the threshold has no impact on sensitivity so $M$ is still $\eps$-DP.

We will argue that, for any $D$, the distributions StrictUDSAboveThreshold$(D)$ and $M(D)$ differ only by $\delta/(1+e^\eps)$ in total variation distance so that we may invoke Lemma \ref{lem:close-to-pure}.

Let $t$ denote the smallest index such that $f_t(D)\geq T$. We make two observations. First, the two algorithms behave identically up to (but not including) $t$. Second, the probability that $M(D)$ halts at some point after $t$ is nonzero but the corresponding StrictUDSAboveThreshold probability is 0; the mass was moved onto $t$ itself. Taking these two observations together, our proof is complete once we prove that the probability that $M(D)$ halts at some point after $t$ is $\leq \delta/(1+e^\eps)$.

It will suffice to dissect the halting condition:
\begin{align*}
f_i(D) +\nu_i \geq{}& \hat{T}\\
f_i(D) -T+2q \geq{}& -\nu_i+\gamma
\end{align*}
where $\gamma$ is the Laplace random variable added to $T$ in $\hat{T}$. Because $f_i(D) -T$, the halting condition would be satisfied when $-\nu_i+\gamma \leq 2q$. Due to symmetry of the Laplace distribution, it suffices to prove $\nu_i+\gamma \leq 2q$ is true except with probability $\frac{\delta}{1+e^\eps}$.

By definition of `quantile,' each Laplace random variable is bounded by $q$ except with probability $\frac{\delta}{2(1+e^\eps)}$. The proof is complete via a union bound.
\end{proof}

\subsection{``Off-the-shelf'' Associative Maps}
We assume there is an off-the-shelf data structure that wraps around a (key, value) array $A$ with the following interface:
\begin{itemize}
    \item GetLoad: returns how many distinct keys are present
    \item GetCapacity: returns how many distinct keys can be stored before a Resize occurs
    \item Resize: lengthens $A$ by a constant factor (e.g., 2) and possibly re-organizes its contents
    \item Present(key): returns True if the key is in $A$
    \item Read(key): throws error if not present. Otherwise, returns the value associated with the present key.
    \item Dump: returns all the (key, value) pairs stored in $A$
    \item Write(key, value): see Algorithm \ref{alg:write}
\end{itemize}

\begin{algorithm}
\caption{map.Write(key, value)}
\label{alg:write}

\If {Present(key)} {

Replace old value with the input value.

} \Else{

Write the input to free space in $A$

}

\If{GetLoad() $\geq$ GetCapacity()}{
    Resize()
}

\Return{}
\end{algorithm}

We can use associative maps to create a histogram out of a stream of elements

\begin{algorithm}
\caption{Histogram algorithm using an associative map}
\KwIn{$x_1,x_2,\dots, x_n$}
\For{$i\in [n]$}{
    value $\gets 1$
    
    \If{map.Present($x_i$)}{
        value $\gets 1+$ map.Read($x_i$)
    }
    
    map.Write($x_i$, value)
}

\Return{map.Dump()}
\end{algorithm}

\subsection{Associative Maps with Private Resizes}
Now we adapt Algorithm \ref{alg:strict-uds-abovethreshold} to our motivating problem: resizing associative maps. At its core, our data structure essentially constructs distinct-element queries of the prefix of the dataset and resizes when the response from StrictUDSAboveThreshold is $\top$. One technical annoyance to handle is ensuring that the impact of one contributor is isolated to one instance of StrictUDSAboveThreshold, but we deal with this by rounding up the count of distinct elements if it is below the preceding threshold.

We add the methods SetNoisedCapacity and GetNoisedCapacity which are simple interfaces to $\hat{T}$. We also add the PrivateWrite method whose pseudocode is presented in Algorithm \ref{alg:private-write}. Finally, we assume there is a method GetPreviousCapacity() to retrieve the capacity that triggered the previous Resize call (0 if it does not exist).

\begin{algorithm}
\caption{map.PrivateWrite(key, value)}
\label{alg:private-write}

\If {Present(key)} {

Replace old value with the input value.

} \Else{

Write the input to free space in $A$

}

AdjustedLoad $\gets \max$(GetPreviousCapacity(),  GetLoad())

$\nu \gets \Lap(2/\eps)$

\If{AdjustedLoad $\geq$ GetCapacity() or AdjustedLoad $+\nu\geq $ GetNoisedCapacity()} {
    Resize()

    SetNoisedCapacity(GetCapacity() $+\Lap(2/\eps) - 2q$) \tcc{$q$ identical to Algorithm \ref{alg:strict-uds-abovethreshold}}
}

\Return{the bit indicating whether Resize was called}
\end{algorithm}

\begin{thm}
Let $M$ denote the algorithm that calls PrivateWrite on each pair in its input in sequence and returns the vector of bits returned by PrivateWrite, indicating when memory resized. Under the assumption that neighboring inputs differ by replacing a single key-value pair, $M$ is $(\eps,\delta)$-DP.
\end{thm}
\begin{proof}
Let $i$ be the index on which the neighbors differ. By construction, the first $i-1$ bits of the output are identically distributed for both inputs. Thus we only need to consider how the distribution of the suffix starting at $i$ changes when we switch between the neighboring inputs.

Once we fix a prefix of output bits, we have fixed a sequence of increasing capacities for the map. So for both neighbors, GetPreviousCapacity() returns the same value starting from index $i$ to the first Resize call. In turn, sensitivity of the `AboveLoad' values is 1. They are also unidirectional because if load is increased at some point by switching the input then it cannot later decrease and applying the max function does not change this property. Hence the timing of the first resize after $i$ is $(\eps,\delta)$-DP via StrictUDSAboveThreshold.

The timing of the subsequent resizes is completely insensitive to the input. The load immediately after a resize is less than the new capacity and each write adds at most one novel key, so the first calculation of AdjustedLoad after the resize is equal to the previous capacity regardless of which input was given. Furthermore, every subsequent calculation has the same outcome regardless of the input because the only point of difference between the input has passed.
\end{proof}
\paragraph{Acknowledgments} We would like to thank Katharine Daly and Stanislav Chiknavaryan for identifying the message length as a potential surface of attack. We would also like to thank Peter Kairouz, Marco Gruteser, and Daniel Ramage for supporting this line of work and encouraging publication. Finally, we would like to thank the rest of the Parfait team for continuing to push the envelope on federated and private computation in industry.

\bibliographystyle{plainnat}
\bibliography{refs.bib}
\newpage
\appendix
\section{Illustration of Signal Conveyed by perf stat}
Below we visualize the number of page faults reported by \texttt{perf stat} when we set up \texttt{std::unordered\_map} to ingest one of two streams of inputs that differ on one key.

\begin{figure}[h]
\centering
    \includegraphics[width=0.65\linewidth]{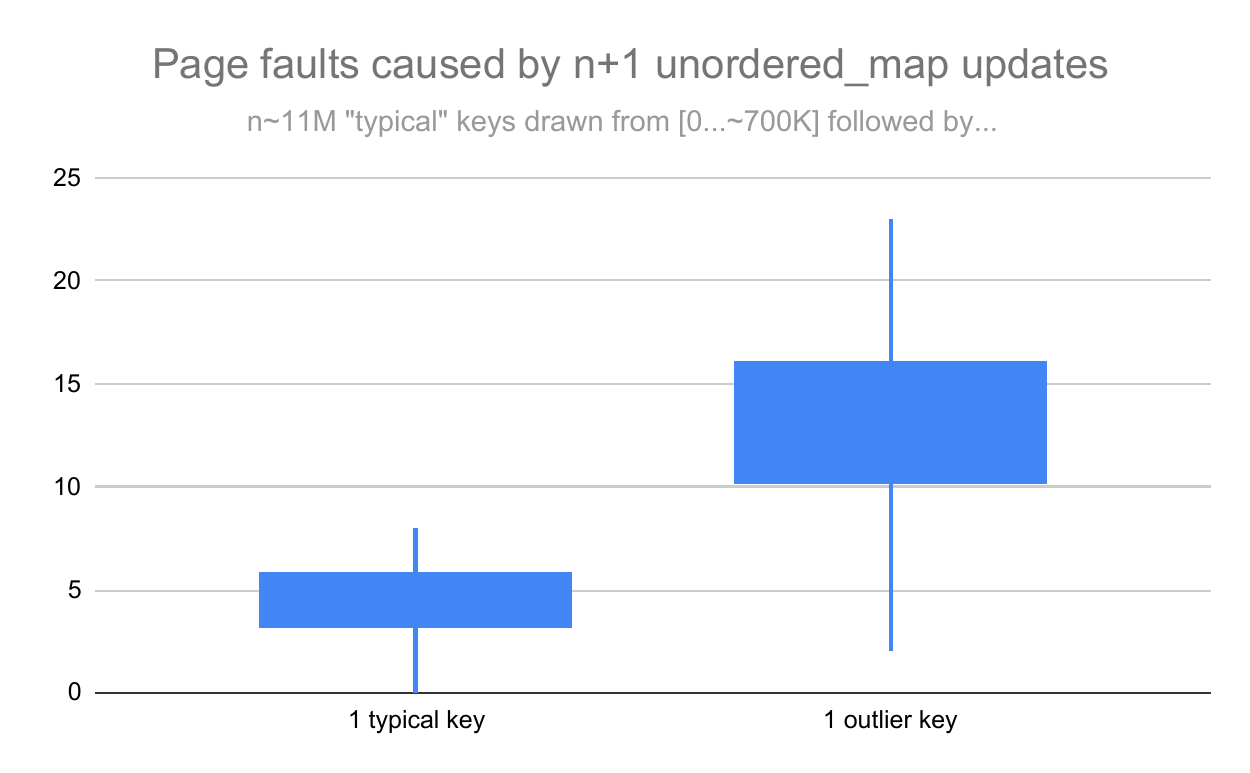}
    \caption{Page faults caused by adding to an \texttt{std::unordered\_map}.}
    \label{fig:perf_attack}
\end{figure}

The \texttt{int,int} map is filled with approximately 11M keys, each drawn from $[0, 712696]$. Every key is mapped to 1. In one case, a final write to key 1 is performed. In another case, the final write is to 712697. The 712696 constant is chosen because it is just below a resizing threshold for \texttt{std::unordered\_map}.
\section{Demonstrating the viability of the Message-Length Attack}
\label{apdx:msg-length-attack}
This section assumes we have pulled the \texttt{tensorflow\_federated} repository prior to change \texttt{11ad3d7}. If we add the following test case to \texttt{dp\_open\_domain\_histogram\_test.cc} and run it, the outcome should be 217. Switching ``iOS'' to ``android'' should change the result to 190.

\begin{verbatim}
TEST(DPOpenDomainHistogramTest, SerializationLengthTest){
  auto intrinsic = CreateIntrinsic2Key2Agg<double, double>(/*epsilon=*/1.0,
    /*delta=*/1e-4, /*l0_bound=*/1);
  auto aggregator = CreateTensorAggregator(intrinsic).value();
  absl::string_view apps[5] = {"Reddit", "Instagram", "X", "TikTok", "Youtube"};

  for (int i = 0; i < 50; i++) {
    Tensor keys1 =
        Tensor::Create(DT_STRING, {1}, CreateTestData<string_view>({apps[i%5]}))
            .value();
    Tensor keys2 =
        Tensor::Create(DT_STRING, {1}, CreateTestData<string_view>({"android"}))
            .value();

    Tensor value_tensor1 =
        Tensor::Create(DT_DOUBLE, {1}, CreateTestData<double>({1.0}))
            .value();
    Tensor value_tensor2 =
        Tensor::Create(DT_DOUBLE, {1}, CreateTestData<double>({1.0}))
            .value();
    auto acc_status =
        aggregator->Accumulate({&keys1, &keys2,
                                &value_tensor1, &value_tensor2});
    EXPECT_THAT(acc_status, IsOk());
  }
  auto serialized_state = std::move(*aggregator).Serialize();
  EXPECT_THAT(serialized_state, IsOk());
  std::cout << "serialized_state size: " << serialized_state.value().size()
            << "\n";
  aggregator =
      DeserializeTensorAggregator(intrinsic, serialized_state.value())
          .value();

  Tensor keys1 =
      Tensor::Create(DT_STRING, {1}, CreateTestData<string_view>({"Reddit"}))
          .value();
  Tensor keys2 =
      Tensor::Create(DT_STRING, {1}, CreateTestData<string_view>({"iOS"}))
          .value();

  Tensor value_tensor1 =
      Tensor::Create(DT_DOUBLE, {1}, CreateTestData<double>({1.0}))
          .value();
  Tensor value_tensor2 =
      Tensor::Create(DT_DOUBLE, {1}, CreateTestData<double>({1.0}))
          .value();
  auto acc_status =
      aggregator->Accumulate({&keys1, &keys2,
                              &value_tensor1, &value_tensor2});
  EXPECT_THAT(acc_status, IsOk());
  auto serialized_state2 = std::move(*aggregator).Serialize();
  EXPECT_THAT(serialized_state2, IsOk());
  std::cout << "serialized_state size: " << serialized_state2.value().size()
            << "\n";
}
\end{verbatim}

\section{Demonstrating the Message-length Mitigation Overhead}
\label{apdx:msg-length-mitigation}
This section assumes we have pulled the latest iteration of \texttt{tensorflow\_federated} repository. If we add the following test case to \texttt{dp\_open\_domain\_histogram\_test.cc} and run it, we will observe the deciles, quartiles, and medians needed to plot Figure \ref{fig:msg-mitigation}.

\begin{verbatim}
TEST(DPOpenDomainHistogramTest, Ablation){
  int distinct_elements[4] = {256, 512, 1024, 2048};
  std::vector<std::string> keys;
  for (int i = 0; i < 2048; i++) {
    std::ostringstream oss;
    oss << std::setw(15) << std::setfill('0') << i;
    keys.push_back(oss.str());
  }

  double epsilons[3] = {0.5, 1, 2};
  int num_runs = 40;
  double delta = 1e-4;
  for (double epsilon : epsilons) {
    for (int num_distinct_elements : distinct_elements) {
      std::vector<int> lengths;
      for (int run = 0; run < num_runs; run++) {
        auto intrinsic = CreateIntrinsic2Key2Agg<double, double>(epsilon, delta,
                                                  /*l0_bound=*/1);
        auto aggregator = CreateTensorAggregator(intrinsic).value();
        for (int i = 0; i < num_distinct_elements; i++) {
          Tensor keys1 =
          Tensor::Create(DT_STRING, {1}, CreateTestData<string_view>({keys[i]}))
              .value();
          Tensor keys2 =
              Tensor::Create(DT_STRING, {1}, CreateTestData<string_view>({"android"}))
                  .value();

          Tensor value_tensor1 =
              Tensor::Create(DT_DOUBLE, {1}, CreateTestData<double>({1.0}))
                  .value();
          Tensor value_tensor2 =
              Tensor::Create(DT_DOUBLE, {1}, CreateTestData<double>({1.0}))
                  .value();
          auto acc_status =
              aggregator->Accumulate({&keys1, &keys2,
                                      &value_tensor1, &value_tensor2});
          EXPECT_THAT(acc_status, IsOk());
        }
        auto serialized_state = std::move(*aggregator).Serialize();
        EXPECT_THAT(serialized_state, IsOk());
        lengths.push_back(serialized_state.value().size());
      }
      // std::cout << epsilon << "," << num_distinct_elements
      //           << "," << lengths[0] << "\n";

      std::sort(lengths.begin(), lengths.end());
      std::cout << absl::StrCat(epsilon, ",",
                                num_distinct_elements, ",",
                                lengths[num_runs / 10], ",",
                                lengths[num_runs / 4], ",",
                                lengths[num_runs / 2], ",",
                                lengths[3 * num_runs / 4], ",",
                                lengths[9 * num_runs / 10], "\n");
    }
  }
}
\end{verbatim}
\section{Bespoke Analysis of Positive Laplace Mechanism}
\label{apdx:positive-laplace}
\begin{thm}
If $\tau := \Delta + q_\delta$, then PositiveLaplaceMechanism is $(\eps, \delta)$-DP.
\end{thm}
Notice that the value of $\tau$ expands to $\Delta+ \frac{\Delta}{\eps}\ln \frac{1}{2\delta}$ but the value in Theorem \ref{thm:simple-positive-laplace} is $\frac{\Delta}{\eps}\ln \frac{e^\eps+1}{2\delta}=\frac{\Delta}{\eps} \ln (e^\eps+1) + \frac{\Delta}{\eps}\ln \frac{1}{2\delta}$.
\begin{proof}
Fix any $v$ and $v'=v+\Delta$. There are three segments of the number line $(-\infty, v)$, $[v, v+\Delta]$, and $(v+\Delta, \infty)$; for any outcome $S\subset \R$, we split it according to those segments. We will bound the leakage incurred on each segment by making separate subclaims. For brevity, we will use $M$ to refer to PositiveLaplaceMechanism

\begin{clm}
\label{clm:zero-mass}
$\pr{}{M(v) \in S\cap (-\infty,v)} = 0 = \pr{}{M(v') \in S\cap (-\infty,v)}$
\end{clm}

\begin{clm}
\label{clm:delta-mass}
$\pr{}{M(v) \in S\cap [v, v+\Delta]}\leq \delta$
\end{clm}

\begin{clm}
\label{clm:prime-delta-mass}
$\pr{}{M(v') \in S\cap [v, v+\Delta]}\leq \delta$
\end{clm}

\begin{clm}
\label{clm:bounded-loss}
Either $S\cap (v+\Delta, \infty)$ is empty or $\exp(-\eps) \leq \frac{\pr{}{M(v) \in S\cap (v+\Delta, \infty)}} {\pr{}{M(v') \in S\cap (v+\Delta, \infty)}} \leq \exp(\eps)$ where numerator and denominator are strictly positive.
\end{clm}

Once we prove the above, we have
\begin{align*}
&\pr{}{M(v)\in S} \\
={}& \pr{}{M(v) \in S\cap (-\infty,v)} + \pr{}{M(v) \in S\cap [v, v+\Delta]} + \pr{}{M(v) \in S\cap (v+\Delta, \infty)}\\
={}& \pr{}{M(v') \in S\cap (-\infty,v)} + \pr{}{M(v) \in S\cap [v, v+\Delta]} + \pr{}{M(v) \in S\cap (v+\Delta, \infty)}\\
\leq{}& \pr{}{M(v') \in S\cap (-\infty,v)} + \pr{}{M(v') \in S\cap [v, v+\Delta]} + \delta + \pr{}{M(v) \in S\cap (v+\Delta, \infty)}\\
\leq{}& \pr{}{M(v') \in S\cap (-\infty,v)} + \pr{}{M(v') \in S\cap [v, v+\Delta]} + \delta + \exp(\eps)\cdot \pr{}{M(v') \in S\cap (v+\Delta, \infty)}\\
\leq{}& \exp(\eps)\cdot \pr{}{M(v') \in S} + \delta
\end{align*}
A completely symmetric series of steps results in $\pr{}{M(v') \in S} \leq e^\eps \pr{}{M(v) \in S} + \delta$

Claim \ref{clm:zero-mass} is immediate from the imposition of $\max(v,\tilde{v})$. To prove Claim \ref{clm:delta-mass}, we rely on the choice of $\tau$:
\begin{align*}
\pr{}{M(v) \in S\cap [v, v+\Delta]} \leq{}& \pr{}{M(v) \leq v+\Delta} \\
={}& \pr{}{\max(v,v+\tau+\eta)\leq v+\Delta} \tag{By construction} \\
={}& \pr{}{\max(0,\tau+\eta)\leq \Delta} \\
={}& \pr{}{\max(0,\Delta+q_\delta+\eta)\leq \Delta} \tag{Choice of $\tau$}\\
={}& \pr{}{q_\delta+\eta \leq 0}\\
={}& \delta \tag{Defn. of $q_\delta$}
\end{align*}

The proof of Claim \ref{clm:prime-delta-mass} is similar:
\begin{align*}
\pr{}{M(v') \in S\cap [v, v+\Delta]} \leq{}& \pr{}{M(v') \leq v+\Delta}\\
={}& \pr{}{\max(v',v'+\tau+\eta)\leq v+\Delta} \tag{By construction} \\
={}& \pr{}{\max(0,\tau+\eta) \leq 0} \tag{Value of $v'$} \\
={}& \pr{}{q_\delta + \eta = 0} \\
<{}& \delta
\end{align*}

We finally prove Claim \ref{clm:bounded-loss}. First, recall that a Laplace distribution places nonzero density on all real values, which means both $\pr{}{M(v) \in S\cap (v+\Delta, \infty)}$ and $\pr{}{M(v') \in S\cap (v+\Delta, \infty)}$ are strictly positive whenever $S\cap (v+\Delta, \infty)$ is non-empty. Then, we reason about the ratio of the two probabilities:
\begin{align*}
\frac{\pr{}{M(v) \in S\cap (v+\Delta, \infty)}} {\pr{}{M(v') \in S\cap (v+\Delta, \infty)}} ={}& \frac{\pr{}{\max(v,v+\tau+\eta) \in S\cap (v+\Delta, \infty)}}{\pr{}{\max(v',v'+\tau+\eta) \in S\cap (v+\Delta, \infty)}}\\
={}&\frac{\pr{}{\max(v,v+\tau+\eta) \in S\cap (v+\Delta, \infty)}}{\pr{}{\max(v+\Delta,v+\Delta+\tau+\eta) \in S\cap (v+\Delta, \infty)}}\\
={}&\frac{\pr{}{\max(v,v+\Delta+q_\delta+\eta) \in S\cap (v+\Delta, \infty)}}{\pr{}{\max(v+\Delta, v+2\Delta+q_\delta +\eta) \in S\cap (v+\Delta, \infty)}}\\
={}&\frac{\pr{}{\max(0,\Delta+q_\delta+\eta) \in S_{-v}\cap (\Delta, \infty)}}{\pr{}{\max(0, \Delta+q_\delta +\eta) \in S_{-v-\Delta}\cap (0, \infty)}}
\end{align*}
where we use $S_{-\Delta}$ to denote the subset of $\R$ where every member is $\Delta$ less than some element of $S$.

We can reinterpret the final ratio as concerning the Laplace distribution with mean $\Delta+q_\delta>0$ and scale parameter $\Delta/\eps$ whose mass below 0 is forced upon 0. Specifically, it compares how much mass the clamped distribution places on a subset of $(\Delta,\infty)$ against the subset formed by shifting values down by $\Delta$. Notice that the ratio would be the same if the $\max(0,\dots)$ clamping was not applied because the mass on the two subsets in question are unaffected. As a result, it suffices to reason about 
$$
\frac{\pr{}{\Delta+q_\delta+\eta \in S_{-v}\cap (\Delta, \infty)}}{\pr{}{\Delta+q_\delta +\eta \in S_{-v-\Delta}\cap (0, \infty)}}
$$
Here, we appeal to the fact that the density placed on any $x$ by a Laplace distribution with scale parameter $\Delta/\eps$ is within $\exp(\eps)$ of the density placed on $x+\Delta$ by the same distribution. This results in an upper and lower bound of $\exp(-\eps),\exp(\eps)$.
\end{proof}

\section{Privacy Proof of UDSAboveThreshold}
\label{apdx:uds}
\begin{proof}[Proof of Thm \ref{thm:uds}]
For any integer $k\geq 1$, we will argue
$$\pr{}{\textrm{UDSAboveThreshold}(D) = \bot^{k-1} \top} \approx_\eps \pr{}{\textrm{UDSAboveThreshold}(D') = \bot^{k-1} \top}$$
only over the randomness of $\hat{T}$ and $\nu_k$: fixing $\nu_1,\dots,\nu_{k-1}$, define the quantity
$$
g(D) := \max_{i<k} f_i(D) +\nu_i
$$
so that
\begin{align*}
&\pr{}{\textrm{UDSAboveThreshold}(D) = \bot^{k-1} \top}\\
={}& \pr{}{\hat{T} > g(D) \textrm{~and~} f_k(D)+\nu_k\geq \hat{T}}\\
={}& \int_{-\infty}^\infty \int_{-\infty}^\infty \indic{t\in (g(D), f_k(D)+v)} \cdot \pr{}{\hat{T}=t}  \pr{}{\nu_k = v} ~ \mathit{dv}~ \mathit{dt}
\end{align*}

We now define the following
\begin{gather*}
\hat{v} := v+g(D)-g(D')+f_k(D')-f_k(D)\\
\hat{t} := t+g(D)-g(D')    
\end{gather*}
in order to perform a change of variables:
\begin{align*}
&\pr{}{\textrm{UDSAboveThreshold}(D) = \bot^{k-1} \top}\\
={}& \int_{-\infty}^\infty \int_{-\infty}^\infty  \pr{}{\hat{T}=\hat{t}}\cdot \pr{}{\nu_k = \hat{v}} \cdot \indic{t+g(D)-g(D') \in (g(D), f_k(D')+v+g(D) -g(D'))} ~ \mathit{dv}~ \mathit{dt}\\
={}& \int_{-\infty}^\infty \int_{-\infty}^\infty  \pr{}{\hat{T}=\hat{t}}\cdot \pr{}{\nu_k = \hat{v}} \cdot \indic{t\in (g(D'), f_k(D')+v)} ~ \mathit{dv}~ \mathit{dt}\\
\leq{}& \int_{-\infty}^\infty \int_{-\infty}^\infty  e^{\eps/2} \pr{}{\hat{T}=t}\cdot e^{\eps/2}\pr{}{\nu_k = v} \cdot \indic{t\in (g(D'), f_k(D')+v)} ~ \mathit{dv}~ \mathit{dt}\\
={}& e^\eps\cdot \pr{}{\hat{T} > g(D') \textrm{~and~} f_k(D')+\nu_k\geq \hat{T}}\\
={}& e^\eps\cdot \pr{}{\textrm{UDSAboveThreshold}(D) = \bot^{k-1} \top}
\end{align*}
where the inequality comes (a) from our choice of Laplace parameter, (b) $|\hat{t}-t|\leq 1$, and (c) $|\hat{v}-v|\leq 1$. (b) follows from the sensitivity bound. To prove (c), we rely on the following technical claim: when $g(D)-g(D')$ is positive (resp. negative) then $f_k(D')-f_k(D)$ must be negative or zero (resp. positive or zero).

If $g(D)-g(D')>0$, then we argue there must be some $i<k$ where $f_i(D)-f_i(D')>0$. Assume for contradiction this is not the case. Then for every $i<k$
\begin{align*}
    f_i(D)-f_i(D') &\leq 0\\
    f_i(D) + \nu_i &\leq f_i(D') + \nu_i
\end{align*}
In turn, we must conclude that $\max_{i<k} f_i(D) + \nu_i \leq \max_{i<k} f_i(D') + \nu_i$. But this contradicts the fact that $g(D)-g(D')>0$.

Because $g(D)-g(D')>0$ implies some $i<k$ where $f_i(D)-f_i(D')>0$, we now invoke unidirectionality of sensitivity: $f_k(D)-f_k(D')\geq 0$. Then we conclude $f_k(D')-f_k(D)$ must be negative or zero as desired.

The proof for the $g(D)-g(D')< 0$ case is symmetric.
\end{proof}

\end{document}